\documentclass{article}[11pt,letter]

 \usepackage[T1]{fontenc}
\usepackage[english]{babel}
\usepackage[utf8]{inputenc}
\usepackage{mathtools}
\usepackage{algorithm}
\usepackage{algorithmicx}
\usepackage[noend]{algpseudocode}
\usepackage{xspace}
\usepackage{amssymb}
\usepackage{amsthm}
\usepackage{mathtools}
\usepackage{authblk}
\usepackage{fullpage}
\usepackage{cite}

\def\LPSP{{\sf LPS}}

\newtheorem{definition}{Definition}[section]
\newtheorem{lemma}[definition]{Lemma}
\newtheorem{theorem}[definition]{Theorem}
\newtheorem{proposition}[definition]{Proposition}

\DeclarePairedDelimiter{\ceil}{\lceil}{\rceil}
\DeclarePairedDelimiter{\floor}{\lfloor}{\rfloor}
\usepackage{xspace}
\newcommand{\midpalin}[1][n]{{\sf midLPS}$_{\Sigma}[#1]$\xspace}
\newcommand{\palin}[1][n]{{\sf LPS}$_{\Sigma}[#1]$\xspace}
\newcommand{\aerr}{\ensuremath{E}}
\newcommand{\alg}{\mathcal{A}}
\newcommand{\dalg}{\mathcal{D}}
\newcommand{\bigo}{\mathcal{O}}

\newcommand{\etal}{et~al.}
\newcommand{\ttl}{\mathit{ttl}}
\usepackage{etoolbox}
\makeatletter
\setbool{@fleqn}{false}
\makeatother

\title{Tight Tradeoffs for Real-Time Approximation of Longest Palindromes in Streams}

\author[1]{Pawe\l{} Gawrychowski}
\author[2]{Oleg Merkurev}
\author[2]{Arseny M. Shur}
\author[3]{Przemys\l{}aw Uzna\'nski}
\affil[1]{Institute of Informatics, University of Warsaw, Poland}
\affil[2]{Institute of Mathematics and Computer Science, Ural Federal University\\
  Ekaterinburg, Russia}
\affil[3]{Department of Computer Science, ETH Z\" urich, Switzerland}

\begin{document}

\maketitle

\begin{abstract}
We consider computing a longest palindrome in the streaming model, where the symbols arrive
one-by-one and we do not have random access to the input.
While computing the answer exactly using sublinear space is not possible in such a setting, one can
still hope for a good approximation guarantee. 
Our contribution is twofold. First, we provide lower bounds on the space requirements for randomized
approximation algorithms processing inputs of length $n$.
We rule out Las Vegas algorithms, as they cannot achieve sublinear space complexity.
For Monte Carlo algorithms, we prove a lower bounds of  $\Omega
( M \log\min\{|\Sigma|,M\})$ bits of memory; here $M=n/E$ for approximating the answer with additive error $E$, and $M= \frac{\log n}{\log (1+\varepsilon)}$ for approximating the answer with multiplicative error $(1 + \varepsilon)$. Second, we design three real-time algorithms for this problem. Our Monte Carlo approximation algorithms for both additive and multiplicative versions of the problem use $\bigo(M)$ words of memory. Thus the obtained lower bounds are asymptotically tight up to a logarithmic factor.
The third algorithm is deterministic and finds a longest palindrome exactly if it is short. This algorithm can be run in parallel with a Monte Carlo algorithm to obtain better results in practice. Overall, both the time and space complexity of finding a longest palindrome in a stream are essentially settled.
\end{abstract}

\section{Introduction}
In the streaming model of computation, a very long input arrives sequentially in small portions and cannot be stored in full due to space limitation. While well-studied in general, this is a rather recent trend in
algorithms on strings. The main goals are minimizing the space complexity, i.e., avoiding storing
the already seen prefix of the string explicitly, and designing real-time algorithm, i.e., processing
each symbol in worst-case constant time. However, the algorithms are usually randomized and return
the correct answer with high probability.
The prime example of a problem on string considered in the streaming model
is pattern matching, where we
want to detect an occurrence of a pattern in a given text.
It is somewhat surprising that one can actually solve it using polylogarithmic space in the streaming model,
as proved by Porat and Porat~\cite{PoPo09}. A simpler solution was later given by Erg{\"u}n
\etal~\cite{EJS10}, while Breslauer
and Galil designed a real-time algorithm~\cite{BrGa11}. Similar questions studied in such setting
include multiple-pattern
matching~\cite{CFPSS15}, approximate pattern matching~\cite{CFPSS16}, and
parametrized pattern matching~\cite{JPS13}.

We consider computing a longest palindrome in the streaming model, where
a palindrome is a fragment which reads the same in both directions.
This is one of the basic questions concerning regularities in texts
and it has been extensively studied in the classical non-streaming setting, see~\cite{ABG95,GaSe78,KMP77,Man75}
and the references therein. The notion of palindromes, but with a slightly different meaning, is
very important in
computational biology, where one considers strings over $\{A,T,C,G\}$ and a palindrome is a sequence
equal to its reverse complement (a reverse complement reverses the sequences and interchanges
$A$ with $T$ and $C$ with $G$); see~\cite{GMN14} and the references therein for
a discussion of their algorithmic aspects. Our results generalize to biological palindromes
in a straightforward manner.

We denote by $\LPSP(S)$ 
the problem of finding the maximum length of a palindrome in a string $S$ (and a starting position
of a palindrome of such length in $S$).
Solving $\LPSP(S)$ in the streaming model was recently considered by Berenbrink
\etal~\cite{BEMS14},
who developed tradeoffs between the bound on the error and the space complexity for additive and multiplicative variants of the problem, that is,
for approximating the length of the longest palindrome with either additive or multiplicative error.
Their algorithms were Monte Carlo, i.e., returned the correct answer with high probability. They also proved
that any Las Vegas algorithm achieving additive error $\aerr$ must necessarily use $\Omega(\frac{n}{\aerr}\log|\Sigma|)$ bits of memory,
which matches the space complexity of their solution up to a logarithmic factor in the $\aerr\in [1,\sqrt{n}]$ range,
but leaves a few questions. Firstly, does the lower bound still hold for Monte Carlo algorithms? Secondly, 
what is the best possible space complexity when $\aerr\in (\sqrt{n},n]$ in the additive variant,
and what about the multiplicative version? Finally, are there real-time algorithms achieving these
optimal space bounds? We answer all these questions.

Our main goal is to settle the space complexity of $\LPSP$. We start with the lower bounds in
Sect.~\ref{section:lowerbounds}. First, we show that Las Vegas algorithms cannot achieve sublinear
space complexity at all.
Second, we prove a lower bound of  $\Omega
( M \log\min\{|\Sigma|,M\})$ bits of memory for Monte Carlo algorithms;
here $M=n/E$ for approximating the answer with
additive error $E$, and $M= \frac{\log n}{\log (1+\varepsilon)}$ for approximating the answer
with multiplicative error $(1 + \varepsilon)$. Then, in Sect.~\ref{section:algorithms}
we design real-time Monte Carlo algorithms matching these lower bounds up to a logarithmic
factor. More precisely, our algorithm for $\LPSP$ with additive error $E\in [1,n]$ uses
$\bigo(n/E)$ words of space, while our algorithms for $\LPSP$ with multiplicative
error $\varepsilon \in (0, 1]$ (resp., $\varepsilon \in (1, n]$) uses $\bigo\big(\frac{\log (n\varepsilon)}{\varepsilon}\big)$ (resp., $\bigo(\frac{\log n}{\log (1+\varepsilon)})$) words of space\footnote{Note that (a) $\log (1+\varepsilon)$ is equivalent to $\varepsilon$ whenever $\varepsilon<1$; (b) the space used by the algorithms is $O(n)$ for any values of errors; (c) the multiplicative lower bound applies to $\varepsilon > n^{-0.99}$, and thus is not contradicting the algorithm space usage.}. 
Finally we present, for any $m$, a deterministic $\bigo(m)$-space real-time algorithm solving $\LPSP$ exactly if the answer is less than $m$ and detecting a palindrome of length $\ge m$ otherwise.
The last result implies that if the input stream is fully random, then with high probability its longest palindrome can be found exactly by a real-time algorithm within logarithmic space.

\paragraph{Notation and Definitions.}
Let $S$ denote a string of length $n$ over an alphabet $\Sigma=\{1,\ldots,N\}$, where $N$ is polynomial in $n$. We write $S[i]$ for the $i$th symbol of $S$ and $S[i..j]$ for its \emph{substring} (or \emph{factor}) $S[i] S[i {+} 1] \cdots S[j]$; thus, $S[1..n]=S$. A \emph{prefix} (resp. \emph{suffix}) of $S$ is a substring of the form $S[1..j]$ (resp., $S[j..n]$). A string $S$ is a \emph{palindrome} if it equals its \emph{reversal} $S[n]S[n{-}1]\cdots S[1]$. By $L(S)$ we denote the length of a longest palindrome which is a factor of $S$. The symbol $\log$ stands for the binary logarithm. 

We consider the \emph{streaming model} of computation: the input string $S[1..n]$ (called the \emph{stream}) is read left to right, one symbol at a time, and cannot be stored, because the available space is sublinear in $n$. The space is counted as the number of $\bigo(\log n)$-bit machine words. An algorithm is \emph{real-time} if the number of operations between two reads is bounded by a constant. An approximation algorithm for a maximization problem has \emph{additive error} $E$ (resp., \emph{multiplicative error} $\varepsilon$) if it finds a solution with the cost at least $OPT-E$ (resp., $\frac{OPT}{1+\varepsilon}$), where $OPT$ is the cost of optimal solution; here both $E$ and $\varepsilon$ can be functions of the size of the input. In the $\LPSP(S)$ problem, $OPT=L(S)$.

A \emph{Las Vegas algorithm} always returns a correct answer, but its working time and memory usage on the inputs of length $n$ are random variables. A \emph{Monte Carlo algorithm} gives a correct answer with high probability (greater than $1-1/n$)  and has deterministic working time and space.

\section{Lower Bounds}
\label{section:lowerbounds}
In this section we use Yao's minimax principle~\cite{Yao77} to prove
lower bounds on the space complexity of the $\LPSP$ problem in the streaming model, where the length $n$ and the alphabet $\Sigma$ of the input stream are specified. We denote this problem by \palin.

\begin{theorem}[Yao's minimax principle for randomized algorithms]
\label{yao}
Let $\mathcal{X}$ be the set of inputs for a problem and $\mathcal{A}$ be the set of all deterministic algorithms solving it. Then,
for any $x \in \mathcal{X}$ and $A \in \mathcal{A}$, the cost of running $A$ on $x$ is denoted by $c(a,x) \ge 0$.

Let $p$ be the probability distribution over $\mathcal{A}$, and let $A$ be an algorithm chosen at random according to $p$. Let $q$ be
the probability distribution over $\mathcal{X}$, and let $X$ be an input chosen at random according to $q$. Then
$\max_{x \in \mathcal{X}} \mathbf{E}[c(A,x)] \ge \min_{a \in \mathcal{A}} \mathbf{E}[c(a,X)].$
\end{theorem}

We use the above theorem for both Las Vegas and Monte Carlo algorithms. For Las Vegas algorithms, we consider only correct algorithms, and $c(x,a)$ is the
memory usage.
For Monte Carlo algorithms, we consider all algorithms (not necessarily correct) with memory usage not exceeding a certain threshold,
and $c(x,a)$ is the correctness indicator function, i.e., $c(x,a)=0$ if the algorithm is correct and $c(x,a)=1$ otherwise.

Our proofs will be based on appropriately chosen padding. The padding requires a constant number
of fresh characters.
If $\Sigma$ is twice as large as the number of required fresh characters, we can still use half of
it to construct a difficult input instance, which does not affect the asymptotics. Otherwise,
we construct a difficult input instance over $\Sigma$, then add enough new fresh characters
to facilitate the padding, and finally reduce the resulting larger alphabet to binary at the expense
of increasing the size of the input by a constant factor.

\begin{lemma}
\label{alphabetreduction}
For any alphabet $\Sigma=\{1,2,\ldots,\sigma\}$ there exists a morphism $h : \Sigma^* \rightarrow \{0,1\}^*$ such that, for any $c\in\Sigma$, $|h(c)| = 2\sigma+6$ and, for any string $w$, $w$ contains a palindrome
of length $\ell$ if and only if $h(w)$ contains a palindrome of length $(2\sigma+6)\cdot \ell$.
\end{lemma}

\begin{proof}
We set:
\[h(c) = 1 1^s 0 1^{s-c} 1 00 1 1^{s-c} 0 1^c 1.\]
Clearly $|h(c)|=2\sigma+6$ and, because every $h(c)$ is a palindrome, if $w$ contains a palindrome
of length $\ell$ then $h(w)$ contains a palindrome of length $(2\sigma+6)\cdot\ell$. Now assume that
$h(w)$ contains a palindrome of length $(2\sigma+6)\cdot\ell$, where $\ell \geq 1$. 
If $\ell=1$ then we obtain that $w$ should contain a palindrome of length $1$, which always holds.
Otherwise, the palindrome contains $00$ inside and we consider two cases.
\begin{enumerate}
\item The palindrome is centered inside $00$. Then it corresponds to an odd palindrome of length
$\ell$ in $w$.
\item The palindrome maps some $00$ to another $00$. Then it corresponds to an even palindrome
of length $\ell$ in $w$.
\end{enumerate}
In either case, the claim holds.
\end{proof}

For the padding we will often use an infinite string $\nu = 0^11^10^21^20^31^3\ldots$, or more precisely
its prefixes of length $d$, denoted $\nu(d)$. Here $0$ and $1$ should be understood as two characters
not belonging to the original alphabet. The longest
palindrome in $\nu(d)$ has length $\bigo(\sqrt{d})$. 

\begin{theorem}[Las Vegas approximation]
\label{th:lv}
Let $\alg$ be a Las Vegas streaming algorithms solving \palin
with additive error $\aerr \le 0.99 n$ or multiplicative error $(1+\varepsilon) \le 100$
using $s(n)$ bits of memory.
Then $\mathbb{E}[s(n)]=\Omega(n \log |\Sigma|)$.
\end{theorem}

\begin{proof}
By Theorem~\ref{yao}, it is enough to construct a probability distribution $\mathcal{P}$ over $\Sigma^n$ such that
for any deterministic algorithm $\dalg$, its expected memory usage on a string chosen according to $\mathcal{P}$ is
$\Omega(n \log |\Sigma|)$ in bits.

Consider solving \palin with additive error $\aerr$.
We define $\mathcal{P}$ as the uniform distribution over $\nu(\frac{\aerr}{2}) x \$\$ y \nu(\frac{\aerr}{2})^R$, where $x,y \in \Sigma^{n'}$, $n' = \frac{n}{2}-\frac{\aerr}{2}-1$, and $\$$ is a special character not in $\Sigma$.
Let us look at the memory usage of $\dalg$ after having read $\nu(\frac{\aerr}{2}) x$. We say that $x$ is ``good'' when the memory usage 
is at most $\frac{n'}{2}\log |\Sigma|$ and ``bad'' otherwise.
Assume that $\frac{1}{2}|\Sigma|^{n'}$ of all $x$'s are good, then there are two strings
$x \not= x'$ such that the state of $\dalg$ after having read both $\nu(\frac{\aerr}{2}) x$ and $\nu(\frac{\aerr}{2}) x'$ is exactly the same. Hence the behavior of $\dalg$ on
$\nu(\frac{\aerr}{2}) x\$\$ x^R \nu(\frac{\aerr}{2})^R$ and $\nu(\frac{\aerr}{2}) x'\$\$ x^R \nu(\frac{\aerr}{2})^R$ is exactly the same. The former is a palindrome of length $n = 2n'+\aerr+2$, so $\dalg$ 
must answer at least  $2n'+2$, and consequently the latter also must contain a palindrome of length at least $2n'+2$.
A palindrome inside $\nu(\frac{\aerr}{2}) x'\$\$ x^R \nu(\frac{\aerr}{2})^R$ is either fully contained within
$\nu(\frac{\aerr}{2})$, $x'$, $x^R$ or it is a middle palindrome.
But the longest palindrome inside $\nu(\frac{\aerr}{2})$ is of length $\bigo(\sqrt{\aerr})<2n'+2$ (for $n$ large enough)
and the longest palindrome inside $x$ or $x^R$ is of length $n'<2n'+2$, so since we have exluced other possibilities,
$\nu(\frac{\aerr}{2}) x'\$\$ x^R \nu(\frac{\aerr}{2})^R$
contains a middle palindrome of length $2n'+2$. This implies that $x=x'$, which is a contradiction.
Therefore, at least $\frac{1}{2}|\Sigma|^{n'}$ of all $x$'s are bad. But then the expected memory usage of $\dalg$ is at least
$\frac{n'}{4}\log |\Sigma|$, which for $\aerr\le0.99 n$ is $\Omega(n \log |\Sigma|)$ as claimed.

Now consider solving \palin with multiplicative error $(1+\varepsilon)$. 
An algorithm with multiplicative error $(1+\varepsilon)$ can also be considered as having additive error $\aerr=n \cdot \frac{\varepsilon}{1+\varepsilon}$, so if the expected memory usage of
such an algorithm is $o(n\log |\Sigma|)$ and $(1+\varepsilon) \le 100$ then we obtain
an algorithm with additive error $\aerr \le 0.99n$ and expected memory usage $o(n\log |\Sigma|)$, which we already know to be impossible.
\end{proof}

Now we move to Monte Carlo algorithms. We first consider exact algorithms solving \palin;
lower bounds on approximation algorithms will be then obtained by padding the input
appropriately. We introduce an auxiliary problem \midpalin, which is to compute the length of the middle palindrome in a string of even length $n$ over an alphabet $\Sigma$. 

\begin{lemma}
\label{lowerbound:exact}
There exists a constant  $\gamma$ such that any randomized Monte Carlo streaming algorithm
$\alg$ solving \midpalin or \palin exactly with probability $1-\frac{1}{n}$
uses at least $\gamma \cdot n \log\min \{|\Sigma|, n\}$ bits of memory.
\end{lemma}

\begin{proof}
First we prove that if $\alg$ is a Monte Carlo streaming algorithm solving \midpalin exactly using less than $\lfloor\frac{n}{2} \log |\Sigma| \rfloor$ bits
of memory, then its error probability is at least $\frac{1}{n|\Sigma|}$.

By Theorem~\ref{yao}, it is enough to construct probability distribution $\mathcal{P}$ over $\Sigma^n$ such that for any deterministic
algorithm $\dalg$ using less than $\lfloor\frac{n}{2} \log |\Sigma| \rfloor$ bits of memory, the expected probability of error on a string chosen
according to $\mathcal{P}$ is at least $\frac{1}{n|\Sigma|}$.

Let $n' = \frac{n}{2}$. For any $x\in\Sigma^{n'}$, $k\in\{1,2,\ldots,n'\}$ and $c\in\Sigma$ we define
\[w(x,k,c) = x[1] x[2] x[3]  \ldots x[n'] x[n'] x[n'-1] x[n'-2] \ldots  x[k+1]  c  0^{k-1}.\]
Now $\mathcal{P}$ is the uniform distribution over all such $w(x,k,c)$.

Choose an arbitrary maximal matching of strings from $\Sigma^{n'}$ into pairs $(x,x')$ such that $\dalg$ is in the same state after reading either $x$ or $x'$. At most one string per state of $\dalg$ is left unpaired, that is at most $2^{\lfloor\frac{n}{2} \log |\Sigma| \rfloor-1}$ strings in total. Since there are $|\Sigma|^{n'}=2^{n' \log |\Sigma|} \ge 2\cdot 2^{\lfloor\frac{n}{2} \log |\Sigma| \rfloor-1}$ possible strings of length $n'$, at least half of the strings are paired.
Let $s$ be longest common suffix of $x$ and $x'$, so $x = v c s$ and $x' = v' c' s$, where $c \not= c'$ are single characters.  Then
$\dalg$ returns the same answer on $w(x,n'-|s|,c)$ and $w(x',n'-|s|,c)$, even though the length of the middle palindrome is exactly $2|s|$ in one of
them, and at least $2|s|+2$ in the other one. Therefore, $\dalg$ errs on at least one of these two inputs. Similarly, it errs on either
$w(x,n'-|s|,c')$ or $w(x,n'-|s|,c')$. Thus the error probability is at least $\frac{1}{2n'|\Sigma|} = \frac1{n|\Sigma|}$. 

Now we can prove the lemma for \midpalin with a standard amplification trick.
Say that we have a~Monte Carlo streaming algorithm, which solves \midpalin exactly
with error probability $\varepsilon$ using $s(n)$ bits of memory. Then we can run its $k$ instances simultaneously and return
the most frequently reported answer. The new algorithm needs $\bigo(k\cdot s(n))$ bits of memory and its error probability $\varepsilon_{k}$ satisfies:
\[\varepsilon_k \le \sum_{2i < k} \binom{k}{i}(1-\varepsilon)^i \varepsilon^{k-i} \le 2^k \cdot \varepsilon^{k/2} = (4 \varepsilon)^{k/2}.\]
Let us choose $\kappa = \frac{1}{6}\frac{\log(4/n)}{\log(1/(n|\Sigma|))} = \frac16 \frac{1-o(1)}{1+\log|\Sigma|/\log n} = \Theta(\frac{\log n}{\log n + \log |\Sigma|}) = \gamma \cdot \frac{1}{\log |\Sigma|} \log \min \{|\Sigma|, n\}$, for some constant $\gamma$.
Now we can prove the theorem. Assume that $\alg$ uses less than $\kappa \cdot n\log |\Sigma| = \gamma \cdot n \log \min  \{|\Sigma|, n\}$ bits of memory. Then running $\left\lfloor \frac{1}{2\kappa} \right\rfloor \ge \frac34\frac1{2\kappa}$ (which holds since $\kappa < \frac16$)
instances of $\alg$ in parallel requires less than $\lfloor \frac{n}{2} \log |\Sigma| \rfloor$ bits of memory. But then 
the error probability of the new algorithm is bounded from above by
\[\left(\frac{4}{n}\right)^{\frac{3}{16\kappa}} = \left(\frac{1}{n|\Sigma|}\right)^{\frac{18}{16}} \le \frac{1}{n|\Sigma|}\]
which we have already shown to be impossible.

The lower bound for \midpalin can be translated into a lower bound for solving \palin exactly
by padding the input so that the longest palindrome is centered in the middle.
Let $x=x[1]x[2]\ldots x[n]$ be the input for \midpalin. We define
\[w(x)= x[1] x[2] x[3]  \ldots x[n/2]  1 \underbracket{000 \ldots 0}_{n} 1 x[n/2+1] \ldots x[n].\]
Now if the length of the middle palindrome in $x$ is $k$, then $w(x)$ contains a palindrome of length
at least $n+k+2$. In the other direction, any palindrome inside $w(x)$ of length $\ge n$ must be centered somewhere
in the middle block consisting of only zeroes and both ones are mapped to each other, so it must be
the middle palindrome.
Thus, the length of the longest palindrome inside $w(x)$ is exactly $n+k+2$, so we have reduced solving \midpalin to solving
\palin[2n+2]. We already know that solving \midpalin[n] with probability $1-\frac{1}{n}$
requires $\gamma \cdot n \log \min \{|\Sigma|, n\}$ bits of memory,
so solving \palin[2n+2] with probability $1-\frac{1}{2n+2}\geq 1-\frac{1}{n}$ requires $\gamma \cdot n \log \{|\Sigma|,n\} \geq \gamma' \cdot (2n+2) \log \min \{|\Sigma|, 2n+2\}$ bits of memory.
Notice that the reduction needs $\bigo(\log n)$ additional bits of memory to count up to $n$, but for large $n$ this is
much smaller than the lower bound if we choose $\gamma' < \frac{\gamma}{4}$.
\end{proof}

To obtain a lower bound for Monte Carlo additive approximation, we observe that any algorithm
solving \palin with additive error $\aerr$ can be used to solve \palin[\frac{n-\aerr}{\aerr+1}] exactly
by inserting $\frac{\aerr}{2}$ zeroes between every two characters, in the very beginning, and in the very end.
However, this reduction requires $\log (\frac{\aerr}{2})\leq \log n$ additional bits of memory for counting up to $\frac{\aerr}{2}$
and cannot be used when the desired lower bound on the required number of bits
$\Omega(\frac{n}{\aerr}\log\min(|\Sigma|,\frac{n}{E})$ is significantly smaller than $\log n$. 
Therefore, we need a separate technical lemma which implies that both additive and multiplicative
approximation with error probability $\frac{1}{n}$ require $\Omega(\log n)$ bits of space.

\begin{lemma}
\label{hashing_lb}
Let $\alg$ be any randomized Monte Carlo streaming algorithm solving \palin with additive error at most $0.99 n$ or multiplicative error at most
$n^{0.49}$  and error probability $\frac{1}{n}$.
Then $\alg$ uses $\Omega(\log n)$ bits of memory.
\end{lemma}

\begin{proof}
By Theorem~\ref{yao}, it is enough to construct a probability distribution $\mathcal{P}$ over $\Sigma^n$, such that for
any deterministic algorithm $\dalg$ using at most $s(n)=o(\log n)$ bits of memory, the expected probability of error on a string chosen according
to $\mathcal{P}$ is $\frac{1}{2^{s(n)+2}}$.

Let $n' = s(n)+1$.  For any $x,y \in \Sigma^{n'}$, let $w(x,y) = \nu(\frac{n}{2}-n')^R x y^R \nu(\frac{n}{2}-n')$.
Observe that if $x=y$ then $w(x,y)$ contains a palindrome of length $n$, and otherwise the longest palindrome there has
length at most $2n'+O(\sqrt{n}) = O(\sqrt{n})$, thus any algorithm with additive error of at most $0.99 n$ or with a multiplicative error at most $n^{0.49}$
must be able to distinguish between these two cases (for $n$ large enough).

Let $S \subseteq \Sigma^{n'}$ be an arbitrary family of strings of length $n'$ such that $|S| = 2 \cdot 2^{s(n)}$, and let $\mathcal{P}$ be 
the uniform distribution on all strings of the form $w(x,y)$, where $x$ and $y$ are chosen uniformly and independently from $S$.
By a counting argument, we can create at least $\frac{|S|}{4}$ pairs $(x,x')$ of elements
from $S$ such that the state of $\dalg$ is the same after having read $\nu(\frac{n}{2}-n')^Rx$ and $\nu(\frac{n}{2}-n')^Rx'$. 
(If we create the pairs greedily, at most one such $x$ per state of memory can be left unpaired, so at least $|S| - 2^{s(n)} = \frac{|S|}2$ elements are paired.)
Thus, $\dalg$ cannot
distinguish between $w(x,x')$ and $w(x,x)$, and between $w(x',x')$ and $w(x',x)$, so its error probability must be at least
$\frac{|S|/2}{|S|^2} = \frac{1}{4\cdot 2^{s(n)}}$. Thus if $s(n) = o(\log n)$, the error rate is at least $\frac{1}{n}$ for $n$ large enough, a contradiction.
\end{proof}

Combining the reduction with the technical lemma and 
taking into account that we are reducing to a problem with string length of $\Theta(\frac{n}{E})$,
we obtain the following.

\begin{theorem}[Monte Carlo additive approximation]
\label{th:montecarlo_additive_lowerbound}
Let $\alg$ be any randomized Monte Carlo streaming algorithm solving \palin with additive error $\aerr$
with probability $1-\frac{1}{n}$. If $\aerr \le 0.99n$ then $\alg$ uses
$\Omega(\frac{n}{\aerr} \log \min \{|\Sigma|,\frac{n}{\aerr}\} )$ bits of memory.
\end{theorem}

\begin{proof}
Define $\sigma = \min \{|\Sigma|,\frac{n}{\aerr}\}.$

Because of Lemma~\ref{hashing_lb}
it is enough to prove that $\Omega(\frac{n}{\aerr} \log \sigma)$ is a lower bound when 
\begin{equation}\label{foo} \aerr \le \frac{\gamma}2 \cdot \frac{n}{\log n} \log \sigma. \end{equation} 
Assume that there is a Monte Carlo streaming algorithm $\alg$ solving \palin with additive error $\aerr$ using $o(\frac{n}{\aerr}\log \sigma)$ 
bits of memory and probability $1-\frac{1}{n}$. 
Let $n' = \frac{n-\aerr/2}{\aerr/2+1} \ge \frac{n}{\aerr}$ (the last inequality, equivalent to $n \ge  \aerr\cdot \frac{\aerr}{\aerr-2}$ holds because $\aerr \le 0.99n$ and because we can assume that $\aerr \ge 200$).  Given a string $x[1] x[2] \ldots x[n']$, we can simulate running $\alg$
on $0^{\aerr} x[1] 0^{\aerr/2} x[2] 0^{\aerr/2} x[3] \ldots 0^{\aerr/2} x[n'] 0^{\aerr/2}$ to calculate $R$ (using $\log (\aerr/2) \leq \log n$ additional bits of memory), and then return $\left\lfloor \frac{R}{\aerr/2+1}\right\rfloor$.
We call this new Monte Carlo streaming algorithm $\alg'$.
Recall that $\alg$ reports the length of the longest palindrome with additive error $\aerr$. Therefore, if the original string contains
a palindrome of length $r$, the new string contains a palindrome of length $\frac{\aerr}{2}\cdot(r+1)+r$, so $R\geq r(\aerr/2+1)$ and $\alg'$
will return at least $r$. In the other direction, if $\alg'$ returns $r$, then the new string contains a palindrome of length $r(\aerr/2+1)$.
If such palindrome is centered so that $x[i]$ is matched with $x[i+1]$ for some $i$, then
it clearly corresponds to a palindrome of length $r$ in the original string. But otherwise every
$x[i]$ within the palindrome is matched with $0$, so in fact the whole palindrome corresponds to a streak
of consecutive zeroes in the new string and can be extended to the left and to the right to start and end
with $0^{\aerr}$, so again it corresponds to a palindrome of length $r$ in the original string.
Therefore, $\alg'$ solves \palin[n'] exactly with probability
$1-\frac{1}{(n'(\aerr/2+1)+\aerr/2)} \ge 1 - \frac{1}{n'}$ and uses
$o(\frac{n'(\aerr/2+1)+\aerr/2}{\aerr/2} \log \sigma)+\log n = o(n' \log \sigma) + \log n$ bits of memory. Observe that by Lemma~\ref{lowerbound:exact} we get a lower bound 
\[
\gamma \cdot n' \log \min \{|\Sigma|,n'\} \geq \frac{\gamma}{2}\cdot n' \log \sigma+ \frac{\gamma}{2}\cdot \frac{n}{E} \log\sigma
\ge \frac{\gamma}{2}\cdot n' \log \sigma + \log n
\label{eq:lb}
\]
 (where the last inequality holds because of Eq.\eqref{foo}). Then, for large $n$
we obtain contradiction as follows
\[
o(n' \log\sigma) + \log n < \frac{\gamma}{2} \cdot n' \log\sigma + \log n. \qedhere
\]
\end{proof}

Finally, we consider multiplicative approximation. 
The proof follows the same basic idea as of Theorem~\ref{th:montecarlo_additive_lowerbound},
however is more technically involved. The main difference is that due to uneven padding, we are
reducing to \midpalin[n'] instead of \palin[n'].

\begin{theorem}[Monte Carlo multiplicative approximation]
\label{th:montecarlo_multiplicative_lowerbound}
Let $\alg$ be any randomized Monte Carlo streaming algorithm solving \palin with multiplicative error $(1+\varepsilon)$
with probability $1-\frac{1}{n}$. If $n^{-0.98} \le \varepsilon \le n^{0.49}$ then $\alg$ uses $\Omega(\frac{\log n}{\log(1+\varepsilon)}\log \min \{|\Sigma|,\frac{\log n}{\log(1+\varepsilon)}\})$ bits of memory.
\end{theorem}

\begin{proof}
For $\varepsilon \ge n^{0.001}$ the claimed lower bound reduces to $\Omega(1)$ bits,
which obviously holds. Thus we can assume that $\varepsilon < n^{0.001}$.
Define \[\sigma = \min \{|\Sigma|, \frac{1}{50} \frac{\log n}{\log(1+2\varepsilon)}-2\}.\]
First we argue that it is enough to prove that $\mathcal{A}$ uses
$\Omega(\frac{\log n}{\log(1+\varepsilon)}\log\sigma)$ bits of memory.
Since $\log (1+2\varepsilon) \le 0.001 \log n + o(\log n)$, we have that:
\begin{equation}
\label{eq_18}
\frac{1}{50}\frac{\log n}{\log(1+2\varepsilon)}-2  \ge 18 - o(1)
\end{equation}
and consequently:
\begin{equation}
\label{eq:sigma1}
\frac{1}{50}\frac{\log n}{\log(1+2\varepsilon)}-2 = \Theta(\frac{\log n}{\log(1+2\varepsilon)}).
\end{equation}
Finally, observe that:
 \begin{equation}
 \label{eq:2eps}
 \log(1+2\varepsilon) = \Theta(\log(1+\varepsilon))
 \end{equation}
because $\log 2(1+\varepsilon) = \Theta(\log(1+\varepsilon))$ for $\varepsilon \ge 1$,
and $\log(1+\varepsilon) = \Theta(\varepsilon)$ for $\varepsilon < 1$.
From \eqref{eq:sigma1} and \eqref{eq:2eps} we conclude that:
\begin{equation}
\label{eq:sigma2}
 \log\sigma = \Theta(\log \min\{|\Sigma|,  \frac{\log n}{\log(1+\varepsilon)} \}).
 \end{equation}

Because of Lemma~\ref{hashing_lb} and equations \eqref{eq:2eps} and \eqref{eq:sigma2},
it is enough to prove that $\Omega(\frac{\log n}{\log(1+\varepsilon)}\log\sigma)$ is a lower bound when
\begin{equation}
\label{foo2}
\log(1+2\varepsilon) \le \gamma \cdot \frac{\log\sigma}{100},
\end{equation} 
as otherwise $\Omega(\frac{\log n}{\log (1+\varepsilon)} \log \sigma ) = \Omega( \frac{\log n}{\log (1+2\varepsilon)} \log \sigma) = \Omega(\log n)$.

Assume that there is a Monte Carlo streaming algorithm $\alg$ solving \palin with multiplicative error $(1+\varepsilon)$
with probability $1-\frac{1}{n}$ using $o(\frac{\log n}{\log(1+\varepsilon)}\log \sigma )$ bits of memory.
Let $x = x[1] x[2] \ldots x[n'] x[n'+1] \ldots x[2n']$ be an input for \midpalin[2n']. We choose $n'$ so that $n=(1+2\varepsilon)^{n'+1} \cdot n^{0.99} $. Then $n'= \log_{(1+2\varepsilon)} (n^{0.01})-1 = \frac{1}{100} \frac{\log n}{\log(1+2\varepsilon)}-1$.
We choose $i_0,i_1, i_{2},i_{3},\ldots,i_{n'}$ so that  $i_0+\ldots+i_d = \lceil(1+2\varepsilon)^{d+1} \cdot n^{0.99}\rceil$ for any $0 \le d \le n'$.

(Observe that for $\varepsilon = \Omega(n^{-0.98})$ we have $i_0 > n^{0.99}$ and $i_1,\ldots,i_d > 2n^{0.01}-1$.) Finally we define:
\[w(x) = \nu(i_{n'})^R x[1] \nu(i_{n'-1})^R  \ldots x[n'] \nu(i_{0})^R \nu(i_{0}) x[n'+1] \nu(i_{1}) \ldots \nu(i_{n'-1}) x[2n'] \nu(i_{n'}).\]

If $x$ contains a middle palindrome of length exactly $2k$, then $w(x)$ contains a middle palindrome of length  $2 (1+2\varepsilon)^{k+1} \cdot n^{0.99}$.
Also, based on the properties of $\nu$, any non-middle centered palindrome in $w(x)$ has length at most $O(\sqrt{n})$, which is less than $n^{0.99}$ for $n$ large 
enough. Since $\lceil 2(1+2\varepsilon)^{k}\cdot n^{0.99}\rceil \cdot (1+\varepsilon) < (2(1+2\varepsilon)^{k}\cdot n^{0.99}+1) \cdot (1+\varepsilon) <2(1+2\varepsilon)^{k+1}\cdot n^{0.99}$, value of $k$ can be extracted from the answer of $\alg$.
Thus, if $\alg$ approximates the middle palindrome in $w(x)$ with multiplicative error $(1+\varepsilon)$ with probability $1-\frac{1}{n}$
using
$o(\frac{\log n}{\log(1+\varepsilon)}\log\sigma)$ bits of memory, we can construct a new algorithm $\alg'$ solving \midpalin[2n'] exactly with
probability $1-\frac{1}{n}>1-\frac{1}{2n'}$ using
\begin{equation}
\label{eq:ub2}
o(\frac{\log n}{\log(1+\varepsilon)}\log\sigma) + \log n
\end{equation}
bits of memory.
By Lemma~\ref{lowerbound:exact}  we get a lower bound
\begin{eqnarray}
\gamma \cdot 2n' \log \min \{|\Sigma|,2n'\}  &=& \frac{\gamma}{50} \cdot \frac{\log n}{\log(1+2\varepsilon)} \log\sigma - 2\gamma\log\sigma\nonumber \\
&\ge &\frac{\gamma}{100} \cdot \frac{\log n}{\log(1+2\varepsilon)} \log\sigma + \log n - 2\gamma\log\sigma
\label{eq:lb2}
\end{eqnarray}
(where the last inequality holds because of~(\ref{foo2})). On the other hand, for large $n$
\begin{eqnarray*}
&\frac{\gamma}{100}\cdot \frac{\log n}{\log(1+2\varepsilon)}\log\sigma  - 2\gamma\log\sigma + \log n=\left(\frac{1}{100}\frac{\log n}{\log(1+2\varepsilon)}-2\right)\gamma\log\sigma +\log n\\
&=\Theta\left(\frac{\log n}{\log(1+\varepsilon)}\log\sigma\right) + \log n
\end{eqnarray*}
so~\eqref{eq:lb2} exceeds~\eqref{eq:ub2}, a contradiction.
\end{proof}

\section{Real-Time Algorithms}
\label{section:algorithms}

In this section we design real-time Monte Carlo algorithms within the space bounds matching the lower bounds from Sect.~\ref{section:lowerbounds} up to a factor bounded by $\log n$. The algorithms make use of the hash function known as the \emph{Karp-Rabin fingerprint} \cite{KaRa87}. Let $p$ be a fixed prime from the range $[n^{3 + \alpha}, n^{4 + \alpha}]$ for some $\alpha> 0$, and $r$ be a fixed integer randomly chosen from $\{1,\ldots, p{-}1 \}$. For a string $S$, its forward hash and reversed hash are defined, respectively, as
\[
\phi^F (S) = \left(\sum\limits_{i = 1}^n  { S[i] \cdot r^{i} }\right) \bmod p\ \text{ and }\ 
\phi^R (S) = \left(\sum\limits_{i = 1}^n  { S[i] \cdot r^{n - i + 1} }\right) \bmod p\,.
\]
Clearly, the forward hash of a string coincides with the reversed hash of its reversal. Thus, if $u$ is a palindrome, then $\phi^F (u) = \phi^R (u)$. The converse is also true modulo the (improbable) collisions of hashes, because for two strings $u\neq v$ of length $m$, the probability that $\phi^F(u) = \phi^F(v)$ is
at most $m/p$.
This property allows one to detect palindromes with high probability by comparing hashes. (This approach is somewhat simpler than the one of \cite{BEMS14}; in particular, we do not need ``fingerprint pairs'' used there.) In particular, a real-time algorithm makes $\bigo(n)$ comparisons and thus faces a collision with probability $\bigo(n^{-1-\alpha})$ by the choice of $p$. All further considerations assume that no collisions happen. For an input stream $S$, we denote $F^F(i,j) = \phi^F(S[i..j])$ and $F^R(i,j) = \phi^R(S[i..j])$. The next observation is quite important.

\begin{proposition}[\cite{BrGa11}] \label{FFR} 
The following equalities hold:
\begin{align*}
F^F(i,j) = &\ r^{ -(i-1) } \left( F^F(1,j) - F^F(1, i {-} 1) \right) \bmod p \,,\\
F^R(i,j) = &\ F^R(1,j) - r^{j-i+1} F^R(1,i{-}1) \bmod p \,.
\end{align*}
\end{proposition}

Let $I(i)$ denote the tuple $(i, F^F(1,i{-}1), F^R(1,i{-}1), r^{-(i - 1)}\bmod p, r^{i}\bmod p)$. The proposition below is immediate from definitions and Proposition~\ref{FFR}.

\begin{proposition} \label{fast}
1) Given $I(i)$ and $S[i]$, the tuple $I(i{+}1)$ can be computed in $\bigo(1)$ time.\\
2) Given $I(i)$ and $I(j{+}1)$, the string $S[i..j]$ can be checked for being a palindrome in $\bigo(1)$ time.
\end{proposition}

\subsection{Additive Error}
\label{ssec:add}

\begin{theorem} \label{add}
There is a real-time Monte Carlo algorithm solving the problem $\LPSP(S)$ with the additive error $E=E(n)$ using $\bigo(n/E)$ space, where $n=|S|$. 
\end{theorem}

First we present a simple (and slow) algorithm which solves the posed problem, i.e., finds in $S$ a palindrome of length $\ell(S)\ge L(S)-E$, where $L(S)$ is the length of the longest palindrome in $S$. Later this algorithm will be converted into a real-time one. We store the sets $I(j)$ for some values of $j$ in a doubly-linked list $SP$ in the decreasing order of $j$'s. The longest palindrome currently found is stored as a pair $answer=(pos,len)$, where $pos$ is its initial position and $len$ is its length. Let $ t_E= \floor {\frac {E}{2}} $.

In Algorithm ABasic we add $I(j)$ to the list $SP$ for each $j$ divisible by $t_E$. This allows us to check for palindromicity, at $i$th iteration, all factors of the form $S[kt_E..i]$. We assume throughout the section that at the beginning of $i$th iteration the value $I(i)$ is stored in a variable $I$.

\begin{algorithm*}
\caption{: Algorithm ABasic, $i$th iteration}
\label{alg:ABasic}
\begin{algorithmic}[1]
\If {$i \bmod t_E = 0$}
\State {add $I$ to the beginning of $SP$}
\EndIf
\State {read $S[i]$; compute $I(i+1)$ from $I$; $I\gets I(i+1)$}
\For {all elements $v$ of $SP$} 
	\If {$S[v.i..i]$ is a palindrome and $answer.len < i{-}v.i{+}1$}
		\State {$answer\gets (v.i, i{-}v.i{+}1)$}
	\EndIf
\EndFor
\end{algorithmic}
\end{algorithm*}

\begin{proposition} \label{abasic}
Algorithm ABasic finds in $S$ a palindrome of length $\ell(S)\ge L(S)-E$ using $\bigo(n/E)$ time per iteration and $\bigo(n/E)$ space.
\end{proposition}

\begin{proof}
Both the time and space bounds arise from the size of the list $SP$, which is bounded by $n/t_E=\bigo(n/E)$; the number of operations per iteration is proportional to this size due to Proposition~\ref{fast}. Now let $S[i..j]$ be a longest palindrome in $S$. Let $k = \ceil[\big]	{		\frac{i}{t_E}	}	t_E$. Then $ i \le k < i + t_E$. At the $k$th iteration, $I(k)$ was added to $SP$; then the palindrome $S[k .. j {-} (k {-} i)]$ was found at the iteration $j - (k - i)$. Its length is
\[
j - (k - i) - k + 1 = j - i - 2 (k - i) + 1 > (j - i + 1) - 2 t _E = L(S) - 2\floor[\Big]	{\frac{E}{2}}	\ge L(S) - E,
\]
as required.
\end{proof}

The resource to speed up Algorithm ABasic stems from the following 

\begin{lemma} \label{a2t}
During one iteration, the length $answer.len$ increases by at most $2 \cdot t_E$.
\end{lemma}

\begin{proof}
Let $S[j .. i]$ be the longest palindrome found at the $i$th iteration. If $i - j + 1 \le 2 t_E$ then the statement is obviously true. Otherwise the palindrome $S[j {+} t_E .. i {-} t_E]$ of length $i - j + 1 - 2 t_E$ was found before (at the $(i {-} t_E)$th iteration), and the statement holds again. 
\end{proof}

Lemma~\ref{a2t} implies that at each iteration $SP$ contains only two elements that can increase $answer.len$. Hence we get the following Algorithm A.

\begin{algorithm*}
\caption{: Algorithm A, $i$th iteration}
\label{alg:A}
\begin{algorithmic}[1]
\If {$i \bmod t_E = 0$}
\State {add $I$ to the beginning of $SP$}
\If {$i = t_E$} 
\State {$sp\gets first(SP)$}
\EndIf
\EndIf
\State {read $S[i]$; compute $I(i+1)$ from $I$; $I\gets I(i+1)$}
\State {$sp\gets previous(sp)$} \Comment {if exists}
\While {$i-sp.i+1 \le answer.len$\ and\ $(sp \ne last(SP))$}
\State {$sp\gets next(sp)$}
\EndWhile
\For {all existing $v$ in $\{sp, next(sp)\}$} 
	\If {$S[v.i..i]$ is a palindrome and $answer.len < i{-}v.i{+}1$}
		\State {$answer\gets (v.i, i{-}v.i{+}1)$}
	\EndIf
\EndFor
\end{algorithmic}
\end{algorithm*}

Due to Lemma~\ref{a2t}, the cycle at lines 9--11 of Algorithm A computes the same sequence of values of $answer$ as the cycle at lines 4--6 of Algorithm ABasic. Hence it finds a palindrome of required length by Proposition~\ref{abasic}. Clearly, the space used by the two algorithms differs by a constant. To prove that an iteration of Algorithm A takes $\bigo(1)$ time, it suffices to note that the cycle in lines 7--8 performs at most two iterations. Theorem~\ref{add} is proved.

\subsection{Multiplicative Error for $\varepsilon \le 1$} \label{ss:mul1}

\begin{theorem} \label{mult}
There is a real-time Monte Carlo algorithm solving the problem $\LPSP(S)$ with multiplicative error $\varepsilon=\varepsilon(n)\in (0,1]$ using $\bigo\big(\frac{\log (n \varepsilon)}{\varepsilon}\big)$ space, where $n=|S|$. 
\end{theorem}

As in the previous section, we first present a simpler algorithm MBasic with non-linear working time and then upgrade it to a real-time algorithm. The algorithm must find a palindrome of length $\ell(S)\ge \frac{L(S)}{1+\varepsilon}$. The next lemma is straightforward.

\begin{lemma} \label{lemell}
If $\varepsilon\in (0,1]$, the condition $\ell(S)\ge L(S)(1-\varepsilon/2)$ implies $\ell(S)\ge \frac{L(S)}{1+\varepsilon}$.
\end{lemma}

We set $q_\varepsilon=\ceil[\big]{\log\frac2{\varepsilon}	}$. The main difference in the construction of algorithms with the multiplicative and additive error is that here all sets $I(i)$ are added to the list $SP$, but then, after a certain number of steps, are deleted from it. The number of iterations the set $I(i)$ is stored in $SP$ is determined by the time-to-live function $\ttl(i)$ defined below. This function is responsible for both the correctness of the algorithm and the space bound.

\begin{algorithm*}
\caption{: Algorithm MBasic, $i$th iteration}
\label{alg:MBasic}
\begin{algorithmic}[1]
\State {add $I$ to the beginning of $SP$}
\For {all $v$ in $SP$} 
\If {$v.i+\ttl(v.i) = i$} 
\State {delete $v$ from $SP$}
\EndIf
\EndFor
\State {read $S[i]$; compute $I(i+1)$ from $I$; $I\gets I(i+1)$}
\For {all $v$ in $SP$} 
	\If {$S[v.i..i]$ is a palindrome and $answer.len < i{-}v.i{+}1$}
		\State {$answer\gets (v.i, i{-}v.i{+}1)$}
	\EndIf
\EndFor
\end{algorithmic}
\end{algorithm*}

Let $\beta(i)$ be the position of the rightmost 1 in the binary representation of $i$ (the position 0 corresponds to the least significant bit). We define
\begin{equation} \label{ttl1}
\ttl(i) = 2^{	q_\varepsilon + 2 + \beta(i)}\,.
\end{equation}
The definition is illustrated by Fig.~\ref{SP}. Next we state a few properties of the list $SP$.

\begin{figure}[!htb]
\centerline{\includegraphics[trim=70 740 137 65,clip]{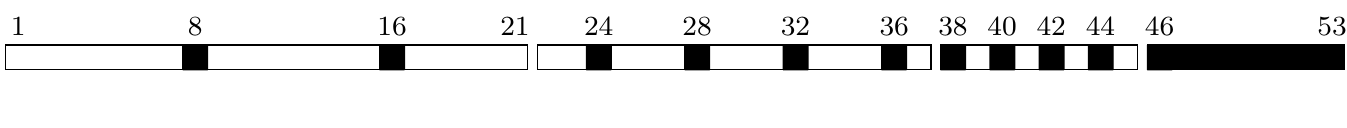} }
\caption{The state of the list $SP$ after the iteration $i=53$ ($q_\varepsilon = 1$ is assumed). Black squares indicate the numbers $j$ for which $I(j)$ is currently stored. For example, \eqref{ttl1} implies $\ttl(28)=2^{1+2+2}=32$, so $I(28)$ will stay in $SP$ until the iteration $28+32=60$.}
\label{SP}
\end{figure}

\begin{lemma} \label{ttl_big}
For any integers $a \ge 1$ and $b \ge 0$, there exists a unique integer $j \in [a, a+2^b)$ such that $\ttl(j) \ge 2^{q_\varepsilon+2+b}$.
\end{lemma}

\begin{proof}
By \eqref{ttl1}, $\ttl(j) \ge 2^{q_\varepsilon + 2 + b}$ if and only if $\beta(j)\ge b$, i.e., $j$ is divisible by $2^b$ by the definition of $\beta$. Among any $2^b$ consecutive integers, exactly one has this property.
\end{proof}

Figure \ref{SP} shows the partition of the range $(0,i]$ into intervals having lengths that are powers of 2 (except for the leftmost interval). In general, this partition consists of the following intervals, right to left:
\begin{equation} \label{intervals}
(i - 2^{q_\varepsilon + 2}, i], (i - 2^{q_\varepsilon + 3}, i - 2^{q_\varepsilon + 2}], \ldots, 
(i - 2^{m}, i - 2^{m-1}],(0, i - 2^{m}], \text{ where } m=\Big\lceil \log \frac{n}{2^{q_\varepsilon + 2}}\Big\rceil - 1.
\end{equation}
Lemma~\ref{ttl_big} and \eqref{ttl1} imply the following lemma on the distribution of the elements of $SP$.

\begin{lemma} \label{SPdistr}
After each iteration, the first interval (resp., the last interval; each of the remaining intervals) in \eqref{intervals} contains $2^{q_\varepsilon + 2}$ (resp., at most $2^{q_\varepsilon + 1}$; exactly $2^{q_\varepsilon + 1}$) elements of the list $SP$.
\end{lemma}

The number of the intervals in \eqref{intervals} is $\bigo(\log (n\varepsilon))$, so from Lemma~\ref{SPdistr} and the definition of $q_\varepsilon$ we have the following.
\begin{lemma} \label{SPsize}
After each iteration, the size of the list $SP$ is $\bigo\big(\frac {\log (n\varepsilon)}{\varepsilon}\big)$.
\end{lemma}

\begin{proposition} \label{cEasyLog}
Algorithm MBasic finds a palindrome of length $\ell(S)\ge \frac{L(S)}{1+\varepsilon}$ using $\bigo(\frac{\log (n\varepsilon)}{\varepsilon})$ time per iteration and $\bigo(\frac{\log (n\varepsilon)}{\varepsilon})$ space.
\end{proposition}
\begin{proof}
Both the time per iteration and the space are dominated by the size of the list $SP$. Hence the required complexity bounds follow from Lemma~\ref{SPsize}. For the proof of correctness, let $S[i..j]$ be a palindrome of length $L(S)$. Further, let $d = \lfloor\log L(S)\rfloor$.

If $d < q_\varepsilon + 2$, the palindrome $S[i..j]$ will be found exactly, because $I(i)$ is in $SP$ at the $j$th iteration:
\[
i + \ttl(i) \ge i + 2^{q_\varepsilon + 2} \ge i + 2^{d + 1} > i + L(S) > j\,.
\]
Otherwise, by Lemma~\ref{ttl_big} there exists a unique $k \in [i, i + 2^{d-q_\varepsilon-1} )$ such that $\ttl(k) \ge 2^{d + 1}$. Hence at the iteration $j - (k - i)$ the palindrome $S[i {+} (k {-} i) .. j {-} (k {-} i)]$ will be found, because $I(k)$ is in $SP$ at this iteration:
\[
k + \ttl(k) \ge i + \ttl(k) \ge i + 2^{d + 1} > j \ge j - (k - i)\,.
\]
The length of this palindrome satisfies the requirement of the proposition:
\[
j - (k - i) - (i + (k - i) )  + 1 = L(S) - 2 (k - i) \ge L(S) - 2^{d - q_\varepsilon} \ge L(S) - \frac{L(S)}{2^{q_\varepsilon}} \ge L(S) \Big(1 - \frac{\varepsilon}{2}\Big) \,.
\]
The reference to Lemma~\ref{lemell} finishes the proof.
\end{proof}

Now we speed up Algorithm MBasic. It has two slow parts: deletions from the list $SP$ and checks for palindromes. Lemmas~\ref{abcd} and~\ref{noMoreThan3} show that, similar to Sect.~\ref{ssec:add}, $\bigo(1)$ checks are enough at each iteration.

\begin{lemma} \label{abcd}
Suppose that at some iteration the list $SP$ contains consecutive elements $I(d), I(c), I(b), I(a)$. Then $b - a \le d - b$.
\end{lemma}

\begin{proof}
Let $j$ be the number of the considered iteration. Note that $a < b < c < d$. Consider the interval in \eqref{intervals} containing $a$. If $a \in (j - 2^{q_\varepsilon + 2}, j]$, then $b - a = 1$ and $d - b = 2$, so the required inequality holds. Otherwise, let $a \in (j - 2^{q_\varepsilon + 2 + x}, j - 2^{q_\varepsilon + 2 + x - 1}]$. Then by \eqref{ttl1} $\beta(a)\ge x$; moreover, any $I(k)$ such that $a<k\le j$ and $\beta(k)\ge x$ is in $SP$. Hence, $b - a \le 2^x$. By Lemma~\ref{SPdistr} each interval, except for the leftmost one, contains at least $2^{q_\varepsilon+1}\ge 4$ elements. Thus each of the numbers $b,c,d$ belongs either to the same interval as $a$ or to the previous interval $(j- 2^{q_\varepsilon + 2 + x - 1}, j - 2^{q_\varepsilon + 2 + x - 2}]$. Again by \eqref{ttl1} we have $\beta(b),\beta(c),\beta(d)\ge x-1$. So $c{-}b,d {-} c \ge 2^{x - 1}$, whence the result.
\end{proof}

We call an element $I(a)$ of $SP$ \emph{valuable at $i$th iteration} if $i-a+1>answer.len$ and $S[a..i]$ can be a palindrome. (That is, Algorithm MBasic does not store enough information to predict that the condition in its line 7 is false for $v=I(a)$.)

\begin{lemma} \label{noMoreThan3}
At each iteration, $SP$ contains at most three valuable elements. Moreover, if $I(d'),I(d)$ are consecutive elements of $SP$ and $i-d'<answer.len\le i-d$, where $i$ is the number of the current iteration, then the valuable elements are consecutive in $SP$, starting with $I(d)$. 
\end{lemma}

\begin{proof}
Let $d$ be as in the condition of the lemma. If $I(d)$ is followed in $SP$ by at most two elements, we are done. If it is not the case, let the next three elements be $I(c), I(b)$, and $I(a)$, respectively. If $S[a..i]$ is a palindrome then $S[a {+} (b {-} a) .. i {-} (b {-} a)]$ is also a palindrome. At the iteration $i {-} (b {-} a)$ the tuple $I(b)$ was in $SP$, so this palindrome was found. Hence, at the $i$th iteration the value $answer.len$ is at least the length of this palindrome, which is $i - a + 1 - 2 (b - a)$. By Lemma~\ref{abcd}, $b-a\le d-b$, implying $answer.len\ge i - a + 1 - (b - a) - (d - b) = i - d + 1$. This inequality contradicts the definition of $d$; hence, $S[a..i]$ is not a palindrome. By the same argument, the elements following $I(a)$ in $SP$ do not produce palindromes as well. Thus, only the elements $I(d),I(c), I(b)$ are valuable.
\end{proof}

Now we turn to deletions. The function $\ttl(x)$ has the following nice property.

\begin{lemma} \label{inject}
The function $x \rightarrow x + \ttl(x)$ is injective.
\end{lemma}
\begin{proof}
Note that $\beta(x + \ttl(x) ) = \beta(x)$ from the definition of $\ttl$. Hence the equality $x + \ttl(x)=y + \ttl(y)$ implies $\beta(x)=\beta(y)$, then $\ttl(x)=\ttl(y)$ by \eqref{ttl1}, and finally $x=y$. 
\end{proof}

Lemma \ref{inject} implies that at most one element is deleted from $SP$ at each iteration. To perform this deletion in $\bigo(1)$ time, we need an additional data structure. By $BS(x)$ we denote a linked list of maximal segments of 1's in the binary representation of $x$. For example, the binary representation of $x=12345$ and $BS(x)$ are as follows:
\[
\arraycolsep=4pt
\begin{array}{|c|c|c|c|c|c|c|c|c|c|c|c|c|c|}
\hline
13 & 12 & 11 & 10 & 9 & 8 & 7 & 6 & 5 & 4 & 3 & 2 & 1 & 0\\
\hline
1 & 1 & 0 & 1 & 0 & 0 & 0 & 0 & 1 & 1 & 1 & 0 & 0 & 1\\
\hline
\end{array}
\hspace*{4mm} BS(12345) = \{[0, 0], [3, 5], [10, 10], [12, 13]\} 
\] 
Clearly, $BS(x)$ uses $\bigo(\log x)$ space.

\begin{lemma}
\label{BSrecalc}
Both $\beta(x)$ and $BS(x + 1)$ can be obtained from $BS(x)$ in $\bigo(1)$ time.
\end{lemma}
\begin{proof}
The first number in $BS(x)$ is $\beta(x)$. Let us construct $BS(x+1)$. Let $[a, b]$ be the first segment in $BS(x)$. If $a > 2$, then $BS(x + 1)=[0,0]\cup BS(x)$. If $a = 1$, then $BS(x + 1)=[0, b]\cup(BS(x)\backslash[1,b])$. Now let $a = 0$. If $BS(x)=\{[0,b]\}$ then $BS(x+1)=\{[b{+}1,b{+}1]\}$. Otherwise let the second segment in $BS(x)$ be $[c, d]$. If $c > b + 2$, then $BS(x + 1)=[b {+} 1, b {+} 1]\cup(BS(x)\backslash [0,b])$. Finally, if $c = b + 2$, then $BS(x + 1)=[b{+}1, d]\cup (BS(x)\backslash \{[0,b],[c,d]\})$.
\end{proof}

Thus, if we support one list $BS$ which is equal to $BS(i)$ at the end of the $i$th iteration, we have $\beta(i)$. If $I(a)$ should be deleted from $SP$ at this iteration, then $\beta(a)=\beta(i)$ (see Lemma~\ref{inject}). The following lemma is trivial.

\begin{lemma} \label{orderTrivial}
If $a < b$ and $\ttl(a) = \ttl(b)$, then $I(a)$ is deleted from $SP$ before $I(b)$.
\end{lemma}

By Lemma~\ref{orderTrivial}, the information about the positions with the same $\ttl$ (in other words, with the same $\beta$) are added to and deleted from $SP$ in the same order. Hence it is possible to keep a queue $QU(x)$ of the pointers to all elements of $SP$ corresponding to the positions $j$ with $\beta(j) = x$. These queues constitute the last ingredient of our real-time Algorithm M.

\begin{algorithm*}
\caption{: Algorithm M, $i$th iteration}
\label{alg:M}
\begin{algorithmic}[1]
\State {add $I$ to the beginning of $SP$}
\If {$i = 1$} 
\State {$sp\gets first(SP)$}
\EndIf
\State {compute $BS[i]$ from $BS$; $BS\gets BS[i]$; compute $\beta(i)$ from $BS$}
\If {$QU(\beta(i) )$ is not empty}
\State {$v\gets \text{ element of }SP\text{ pointed by }first(QU(\beta(i) ) )$}
\If {$v = sp$}
\State {$sp\gets next(sp)$}
\EndIf
\State {delete $v$; delete $first(QU(\beta(i) ) )$}
\EndIf
\State {add pointer to $first(SP)$ to $QU(\beta(i) )$}
\State {read $S[i]$; compute $I(i+1)$ from $I$; $I\gets I(i+1)$}
\State {$sp\gets previous(sp)$} \Comment{if exists}
\While {$i-sp.i+1 \le answer.len$ and $sp \ne last(SP)$}
\State {$sp\gets next(sp)$}
\EndWhile
\For {all existing $v$ in $\{sp, next(sp), next(next(sp) )\}$} 
	\If {$S[v.i..i]$ is a palindrome and $answer.len < i{-}v.i{+}1$}
		\State {$answer\gets (v.i, i{-}v.i{+}1)$}
	\EndIf
\EndFor
\end{algorithmic}
\end{algorithm*}

\begin{proof}[Proof of Theorem~\ref{mult}]
After every iteration, Algorithm M has the same list $SP$ (see Fig.~\ref{SP}) as Algorithm MBasic, because these algorithms add and delete the same elements. Due to Lemma~\ref{noMoreThan3}, Algorithm M returns the same answer as Algorithm MBasic. Hence by Proposition~\ref{cEasyLog} Algorithm M finds a palindrome of required length. 
Further, Algorithm M supports the list $BS$ of size $\bigo(\log n)$ and the array $QU$ containing $\bigo(\log n)$ queues of total size equal to the size of $SP$. Hence, it uses $\bigo(\frac{\log (n\varepsilon)}{\varepsilon})$ space in total by Lemma~\ref{SPsize}.
The cycle in lines 13--14 performs at most three iterations. Indeed, let $z$ be the value of $sp$ after the previous iteration. Then this cycle starts with $sp=previous(z)$ (or with $sp=z$ if $z$ is the first element of $SP$) and ends with $sp=next(next(z))$ at the latest. By Lemma~\ref{BSrecalc}, both $BS(i)$ and $\beta(i)$ can be computed in $\bigo(1)$ time. Therefore, each iteration takes $\bigo(1)$ time.
\end{proof}

\paragraph{Remark}
Since for $n^{-0.99} \le \varepsilon\le 1$ the classes $\bigo\big(\frac{\log n}{\log(1+\varepsilon)}\big)$ and $\bigo\big(\frac{\log (n\varepsilon)}{\varepsilon}\big)$ coincide, Algorithm M uses space within a $\log n$ factor from the lower bound of Theorem~\ref{th:montecarlo_multiplicative_lowerbound}. Furthermore, for an arbitrarily slowly growing function $\varphi$ Algorithm M uses $o(n)$ space whenever $\varepsilon=\frac{\varphi(n)}{n}$.

\subsection{Multiplicative Error for $\varepsilon > 1$}

\begin{theorem} \label{mult2}
There is a real-time Monte Carlo algorithm solving the problem $\LPSP(S)$ with multiplicative error $\varepsilon=\varepsilon(n)\in (1,n]$ using $\bigo\big(\frac{\log (n)}{\log(1+\varepsilon)}\big)$ space, where $n=|S|$. 
\end{theorem}

We transform Algorithm M into real-time Algorithm M' which solves $\LPSP(S)$ with the multiplicative error $\varepsilon>1$ using $\bigo\big(\frac{\log n}{\log(1+\varepsilon)}\big)$ space. The basic idea of transformation is to replace all binary representations with those in base proportional to $1+\varepsilon$, and thus shrink the size of the lists $SP$ and $BS$.

First, we assume without loss of generality that $\varepsilon \ge 7$, as otherwise we can fix $\varepsilon = 1$ and apply Algorithm~M.
Fix $k \le \frac{1}{2}(1+\varepsilon)$ as the largest such even integer (in particular, $k \ge 4$). Let $\beta'(i)$ be the position of the rightmost non-zero digit in the $k$-ary representation of $i$. We define
\begin{equation} \label{ttl'}
\ttl'(i) = 
\begin{cases}
\frac92 \cdot k^ {\beta'(i)} \quad \text{ if } \beta'(i) > 0,\\
4 \quad \text{ otherwise.}
\end{cases}
\end{equation}

The list $SP'$ is the analog of the list $SP$ from Sect.~\ref{ss:mul1}. It contains, after $i$th iteration, the tuples $I(j)$ for all positions $j\le i$ such that $j+ttl'(j)>i$. Similar to \eqref{intervals}, we partition the range $(0; i]$ into intervals and then count the elements of $SP'$ in these intervals. The intervals are, right to left,
\begin{equation} \label{int'}
(i - 4, i], \big(i - \tfrac92 k, i - 4\big], \big(i - \tfrac92k^2, i - \tfrac92k\big], \ldots,
\big(i - \tfrac92k^m, i - \tfrac92k^{m-1}\big], \big(0, i - \tfrac92k^m\big].
\end{equation}

For convenience, we enumerate the intervals starting with 0.

\begin{lemma}
\label{lem:modified_intervals}
Each interval in \eqref{int'} contains at most 5 elements of $SP'$. Each of the intervals $0,...,m$ contains at least 3 elements of $SP'$.
\end{lemma}
\begin{proof}
The $0$th interval contains exactly 4 elements. For any $j=1,\ldots,m{+}1$, an element $x$ of the $j$th interval is in $SP'$ if and only if its position is divisible by $k^j$; see \eqref{ttl'}. The length of this interval is less than $\frac92 k^j$, giving us upper bound of $\big\lceil \frac92 \big\rceil = 5$ elements. Similarly, if $j\ne m{+}1$, the $j$th interval has the length $\frac{9}{2} k^j - \frac{9}{2} k^{j-1}$ and thus contains at least $\big\lfloor \frac{9}{2} \frac{k-1}{k} \big\rfloor$ elements of $SP'$. Since $k \ge 4$, the claim follows.
\end{proof}

Next we modify Algorithm MBasic, replacing $ttl$ by $ttl'$ and $SP$ by $SP'$.

\begin{proposition}
Modified Algorithm MBasic finds a palindrome of length $\ell(S) \ge \frac{L(S)}{1+\varepsilon}$ using $\bigo\big(\frac{\log n}{\log(1+\varepsilon)}\big)$ space.
\end{proposition}
\begin{proof}
Let $S[i..j]$ be a palindrome of length $L(S)$. Let $d = \big\lfloor \log_k \frac{L(S)}{4} \big\rfloor$. Without loss of generality we assume $d \ge 0$, as otherwise $L(S) < 4\le \ttl'(i)$ and the palindrome $S[i..j]$ will be detected exactly. Since $L(S) \ge 4 k^d$, let $a_1<a_2<a_3<a_4<a_5$ be consecutive positions which are multiples of $k^d$ (i.e., $\beta'(a_1),\ldots,\beta'(a_5) \ge d$) such that $a_2 \le \frac{i+j}2 < a_3$. Then in particular $i < a_1$, and there is a palindrome $S[ a_1 .. (i+j-a_1) ]$ such that $a_3 \le (i+j-a_1) < a_5$. Since $a_1 + \ttl'(a_1) \ge a_5$, this particular palindrome will be detected by the modified Algorithm MBasic; thus
$\ell(S) \ge a_3-a_1=2 k^d$. However, we have $L(S) < 4k^{d+1}$, hence $\frac{L(S)}{\ell(S)} < 2k \le (1+\varepsilon)$.

Space complexity follows from bound on size of the list $SP'$, which is at most $5 \big\lceil \frac{\log n}{\log k} \big\rceil = \bigo\big( \frac{\log n}{\log (1+\varepsilon)}\big)$.
\end{proof}

To follow the framework from the case of $\varepsilon \le 1$, we provide analogous speedup to the checks for palindromes. We adopt the same notion of an element valuable at $i$th iteration as in Sect.~\ref{ss:mul1}. First we need the following property, which is a more general analog of Lemma~\ref{abcd}; an analog of Lemma~\ref{noMoreThan3} is then proved with its help. 

\begin{lemma} \label{l:abc}
Suppose that at some iteration the list $SP'$ contains consecutive elements $I(d),I(c)$, and $d\le i-answer.len$, where $i$ is the number of the current iteration. Further, let $I(a)$ be another element of $SP'$ at this iteration and $a<c$. If $c,d$ belong to the same interval of \eqref{int'}, then $I(a)$ is not valuable.
\end{lemma}
\begin{proof}
Let $c,d$ belong to the $j$th interval. Thus they are divisible by $k^j$ and $d-c=k^j$. Since $a<c$, $a$ is divisible by $k^j$ as well. One of the numbers $\frac{d+a}2, \frac{c+a}2$ is divisible by $k^j$; take it as $b$. Let $\delta=b-a$. If $S[a..i]$ is a palindrome, then $S[b..i-\delta]$ is also a palindrome. Since at the $i$th iteration the left border of the $j$th interval was smaller than $c$, then at the $(i-\delta)$th iteration this border was smaller than $b$; hence, $I(b)$ was in $SP'$ at that iteration, and the palindrome $S[b..i-\delta]$ was found. Its length is 
$$
i-\delta-b+1=i+1+a-2b\ge i+1+a-(d+a)\ge i-d+1> answer.len,
$$
which is impossible by the definition of $answer.len$. So $S[a..i]$ is not a palindrome, and the claim follows.
\end{proof}

\begin{lemma}
At each iteration, $SP'$ contains at most three valuable elements. Moreover, if $I(d'),I(d)$ are consecutive elements of $SP'$ and $i-d'<answer.len\le i-d$, where $i$ is the number of the current iteration, then the valuable elements are consecutive in $SP'$, starting with $I(d)$. 
\end{lemma}

\begin{proof}
Let $a<b<c<d$ be such that the elements $I(d),I(c),I(b)$ are consecutive in $SP'$ and $I(a)$ belongs to $SP'$. Then either $b,c$ or $c,d$ are in the same interval of \eqref{int'}, and thus $a$ is not valuable by Lemma~\ref{l:abc}.
\end{proof}

To complete the proof, we now turn to deletions, proving the following analog of Lemma~\ref{inject}.

\begin{lemma}
The function $h(x)= x + \ttl'(x)$ maps at most two different values of $x$ to the same value. Moreover, if $h(x)=h(y)$ and $\beta'(x)\ge \beta'(y)$, then $\beta'(x)=\beta'(h(x))+1$ and $\beta'(y)=0$.
\end{lemma}
\begin{proof}
Let $h(x)=h(y)$. If $\beta'(x)=\beta'(y)$ then $ttl'(x)=ttl'(y)$ by \eqref{ttl'}, implying $x=y$. Hence all preimages of $h(x)$ have different values of $\beta'$. Assume $\beta'(x)>\beta'(y)$. Then we have, for some integer $j$, $x=j\cdot k^{\beta'(x)}$ and $h(x)=(j+4)k^{\beta'(x)} + \frac k2\cdot k^{\beta'(x)-1}$ by \eqref{ttl'}. Since $k$ is even, we get $\beta'(h(x))=\beta'(x)-1$. If $\beta'(y)>0$, we repeat the same argument and obtain $\beta'(x)=\beta'(y)$, contradicting our assumption. Thus $\beta'(y)=0$. The claim now follows.
\end{proof}

We also define a list $BS'(x)$, which maintains an RLE encoding of the $k$-ary representation of $x$. The list $BS'(x)$ has length $\bigo\big(\frac {\log n}{\log k}\big)$, can be updated to $BS'(x{+}1)$ in $\bigo(1)$ time, and provides the value $\beta'(x)$ in $\bigo(1)$ time also (cf. Lemma~\ref{BSrecalc}). Further, Lemma~\ref{orderTrivial} holds for the function $ttl'$, so we introduce the queues $QU'(x)$ in the same way as the queues $QU(x)$ in Sect.~\ref{ss:mul1}. Having all the ingredients, we present Algorithm~M' which speeds up the modified Algorithm MBasic and thus proves Theorem~\ref{mult2}. The only significant difference between Algorithm~M and Algorithm~M' is in the deletion of tuples from the list (compare lines 5--9 of Algorithm~M against lines 5--15 of Algorithm~M').

\begin{algorithm*}
\caption{: Algorithm M', $i$th iteration}
\label{alg:M'}
\begin{algorithmic}[1]
\State {add $I$ to the beginning of $SP'$}
\If {$i = 1$} 
\State {$sp\gets first(SP')$}
\EndIf
\State {compute $BS'[i]$ from $BS'$; $BS'\gets BS'[i]$; compute $\beta(i)$ from $BS'$}
\If {$QU'(\beta(i) + 1)$ is not empty}
\State {$v\gets \text{ element of }SP'\text{ pointed by }first(QU'(\beta(i) ) )$}
\If {$v.i+ttl'(v.i) = i$}
\If {$v = sp$}
\State {$sp\gets next(sp)$}
\EndIf
\State {delete $v$; delete $first(QU'(\beta(i) + 1 ) )$}
\EndIf
\EndIf
\State {$v\gets \text{ element of }SP'\text{ pointed by }first(QU'(0 ) )$}
\If {$v.i+ttl'(v.i) = i$}
\If {$v = sp$}
\State {$sp\gets next(sp)$}
\EndIf
\State {delete $v$; delete $first(QU'(0 ) )$}
\EndIf
\State {add pointer to $first(SP')$ to $QU'(\beta(i) )$}
\State {read $S[i]$; compute $I(i+1)$ from $I$; $I\gets I(i+1)$}
\State {$sp\gets previous(sp)$} \Comment{if exists}
\While {$i-sp.i+1 \le answer.len$ and $sp \ne last(SP')$}
\State {$sp\gets next(sp)$}
\EndWhile
\For {all existing $v$ in $\{sp, next(sp), next(next(sp) )\}$} 
	\If {$S[v.i..i]$ is a palindrome and $answer.len < i{-}v.i{+}1$}
		\State {$answer\gets (v.i, i{-}v.i{+}1)$}
	\EndIf
\EndFor
\end{algorithmic}
\end{algorithm*}

\subsection{The Case of Short Palindromes}

A typical string contains only short palindromes, as Lemma~\ref{longpal} below shows (for more on palindromes in random strings, see \cite{RuSh16b}). Knowing this, it is quite useful to have a deterministic real-time algorithm which finds a longest palindrome exactly if it is ``short'', otherwise reporting that it is ``long''. The aim of this section is to describe such an algorithm (Theorem~\ref{exact}).

\begin{lemma} \label{longpal}
If an input stream $S\in\Sigma^*$ is picked up uniformly at random among all strings of length $n$, where $n\ge|\Sigma|$, then for any positive constant $c$ the probability that $S$ contains a palindrome of length greater than $\frac{2(c+1)\log n}{\log|\Sigma|}$ is $\bigo(n^{-c})$. 
\end{lemma}

\begin{proof}
A string $S$ contains a palindrome of length greater than $m$ if and only if $S$ contains a palindrome of length $m{+}1$ or $m{+}2$. The probability $P$ of containing such a palindrome is less than the expected number $M$ of palindromes of length $m{+}1$ and $m{+}2$ in $S$. A factor of $S$ of length $l$ is a palindrome with probability $1/|\Sigma|^{\lfloor l/2\rfloor}$; by linearity of expectation, we have
$$
M=\frac{n-m}{|\Sigma|^{\lfloor (m+1)/2\rfloor}}+\frac{n-m-1}{|\Sigma|^{\lfloor (m+2)/2\rfloor}}\,.
$$ 
Substituting $m=\frac{2(c+1)\log n}{\log|\Sigma|}$, we get $M=\bigo(n^{-c})$, as required.
\end{proof}

\begin{theorem} \label{exact}
Let $m$ be a positive integer. There exists a deterministic real-time algorithm working in $\bigo(m)$ space, which\\
- solves $\LPSP(S)$ exactly if $L(S)<m$;\\
- finds a palindrome of length $m$ or $m{+}1$ as an approximated solution to $\LPSP(S)$ if $L(S)\ge m$.
\end{theorem}

\begin{proof}
To prove Theorem~\ref{exact}, we present an algorithm based on the Manacher algorithm \cite{Man75}. We add two features: work with a sliding window instead of the whole string to satisfy the space requirements and lazy computation to achieve real time. (The fact that the original Manacher algorithm admits a real-time version was shown by Galil \cite{Gal76}; we adjusted Galil's approach to solve $\LPSP$.) The details follow. 

We say that a palindromic factor $S[i..j]$ has \emph{center} $\frac{i+j}2$ and \emph{radius} $\frac{j-i}2$. Thus, odd-length (even-length) palindromes have integer (resp., half integer) centers and radiuses. This looks a bit weird, but allows one to avoid separate processing of these two types of palindromes. Manacher's algorithm computes, in an online fashion, an array of maximal radiuses of palindromes centered at every position of the input string $S$. A variation, which outputs the length $L$ of the longest palindrome in a string $S$, is presented as Algorithm EBasic below. This variation is similar to the one of \cite{KRS13}. Here, $n$ stays for the length of the input processed so far, $c$ is the center of the longest suffix-palindrome of the processed string. The array of radiuses $Rad$ has length $2n{-}1$ and its elements are indexed by all integers and half integers from the interval $[1,n]$. Initially, $Rad$ is filled with zeroes. The left endmarker is added to the string for convenience. After each iteration, the following invariant holds: the element $Rad[i]$ has got its true value if $i<c$ and equals zero if $i>c$; the value $Rad[c]=n-c$ can increase at the next iteration. Note that the longest palindrome in $S$ coincides with the longest suffix-palindrome of $S[1..i]$ for some $i$. At the moment when the input stream ends, the algorithm has already found all such suffix-palindromes, so it can stop without filling the rest of the array $Rad$.

\begin{algorithm*}
\caption{: Algorithm EBasic, $i$th iteration}
\label{alg:EBasic}
\begin{algorithmic}[1]
\Procedure{Manacher}{}
\State{$c\gets 1$; $L\gets 1$; $n\gets 1$; $S[0]\gets${'\#'}}
\While {not (end of input)}
\State{read($S[n+1]$); {\sc	AddLetter}($S[n+1]$)}
\If {$2*Rad[c]+1>L$} 
\State{$L\gets 2*Rad[c]+1$}
\EndIf
\EndWhile
\State{return $L$}		
\EndProcedure
\Procedure{AddLetter}{$a$}
\State{$s\gets c$}
\While {$c < n+1$}
\State{$Rad[c]\gets min(Rad[2*s-c], n-c) $}
\If {$c+Rad[c] = n$ and $S[c-Rad[c]-1] = a$}
\State{$Rad[c]\gets Rad[c]+1$} 
\State{break} \Comment{longest suf-pal of $S[1..n+1]$ is found}
\EndIf
\State{$c\gets c+0.5$}	\Comment{next candidate for the center}
\EndWhile
\State{$n\gets n+1$}
\EndProcedure
\end{algorithmic}
\end{algorithm*}

Note that $n$ calls to AddLetter perform at most $3n$ iterations of the cycle in lines 10--15 (each call performs the first iteration plus zero or more ``additional'' iterations; the value of $c$ gets increased before each additional iteration and never decreases). So, Algorithm EBasic works in $\bigo(n)$ time but not in real time; for example, reading the last letter of the string $a^nb$ requires $n$ iterations of the cycle.

By conditions of the theorem, we are not interested in palindromes of length $>m{+}1$. Thus, processing a suffix-palindrome of length $m$ or $m{+}1$ we assume that the symbol comparison in line 12 fails. So the procedure AddLetter needs no access to $S[i]$ or $Rad[i]$ whenever $i<n-m$. Hence we store only recent values of $S$ and $Rad$ and use circular arrays $CS$ and $CRad$ of size $\bigo(m)$ for this purpose. For example, the symbol $S[n{-}i]$ is stored in $CS[(n{-}i)\bmod (m{+}1)]$ during $m{+}1$ successful iterations of the outer cycle (lines 3--6), and then is replaced by $S[n{-}i{+}m{+}1]$; the same scheme applies to the array $Rad$. In this way, all elements of $S$ and $Rad$, needed by Algorithm EBasic, are accessible in constant time. Further, we define a queue $Q$ of size $q$ for lazy computations; it contains symbols that are read from the input and await processing. 

Now we describe real-time Algorithm E. It reads input symbols to $Q$ and stops when the end of the input is reached. After reading a symbol, it runs procedure Manacher, requiring the latter to pause after three ``inner iterations'' (of the cycle in lines 10--15). The procedure reads symbols from $Q$; if it tries to read from the empty queue, it pauses. When the procedure is called next time, it resumes from the moment it was stopped. Algorithm E returns the value of $L$ after the last call to procedure Manacher. 

To analyze Algorithm E, consider the value $X=q+n-c$ after some iteration (clearly, this iteration has number $q{+}n$) and look at the evolution of $X$ over time. Let $\Delta f$ denote the variation of the quantity $f$ at one iteration. Note that $\Delta(q{+}n)=1$. Let us describe $\Delta X$. First assume that the iteration contains three inner iterations. Then $\Delta n=0,1,2$ or 3 and, respectively, $\Delta c=1.5,1,0.5$ or 0. Hence 
$$
\Delta X=1-\Delta c= 1+(\Delta n -3)/2=1-(1-\Delta q -3)/2=-(\Delta q)/2. 
$$
If the number of inner iterations is one or two, then $q$ becomes zero (and was 0 or 1 before this iteration); hence $\Delta n=1-\Delta q\ge 1$. Then $\Delta c\le 0.5$ and finally $\Delta X > 0$. From these conditions on $\Delta X$ it follows that
\begin{itemize}
\item[($*$)] if the current value of $q$ is positive, then the current value of $X$ is less than the value of $X$ at the moment when $q$ was zero for the last time.
\end{itemize}
Let $X'$ be the previous value of $X$ mentioned in ($*$). Since the difference $n-c$ does not exceed the radius of some palindrome, $X'\le m/2$. Since $q\le X<X'$, the queue $Q$ uses $\bigo(m)$ space. Therefore the same space bound applies to Algorithm E.

It remains to prove that Algorithm E returns the same number $L$ as Algorithm EBasic with a sliding window, in spite of the fact that Algorithm E stops earlier (the content of $Q$ remains unprocessed by procedure Manacher). Suppose that Algorithm E stops with $q>0$ after $n$ iterations. Then the longest palindrome that could be found by processing the symbols in $Q$ has the radius $X=n+q-c$. Now consider the iteration mentioned in ($*$) and let $n'$ and $c'$ be the values of $n$ and $c$ after it; so $X'=n'-c'$. Since $q$ was zero after that iteration, procedure Manacher read the symbol $S[n']$ during it; hence, it tried to extend a suffix-palindrome of $S[1..n'{-}1]$ with the center $c''\le c'$. If this extension was successful, then a palindrome of radius at least $X'$ was found. If it was unsuccessful, then $c'\ge c''+1/2$ and hence $S[1..n'{-}1]$ has a suffix-palindrome of length at least $X'-1/2$. Thus, a palindrome of length $X\le X'-1/2$ is not longer than a longest palindrome seen before, and processing the queue cannot change the value of $L$. Thus, Algorithm E is correct.

The center and the radius of the longest palindrome in $S$ can be updated each time the inequality in line 6 of procedure Manacher holds. Theorem~\ref{exact} is proved.
\end{proof}

\paragraph{Remark}
Lemma~\ref{longpal} and Theorem~\ref{exact} show a practical way to solve $\LPSP$. Algorithm E is fast and lightweight ($2m$ machine words for the array $Rad$, $m$ symbols in the sliding window and at most $m$ symbols in the queue; compare to 17 machine words per one tuple $I(i)$ in the Monte Carlo algorithms). So it makes direct sense to run Algorithm M and Algorithm E, both in $\bigo(\log n)$ space, in parallel. Then either Algorithm E will give an exact answer (which happens with high probability if the input stream is a ``typical'' string) or both algorithms will produce approximations: one of fixed length and one with an approximation guarantee (modulo the hash collision).

\bibliography{my_bib}

\end{document}